\documentclass[12pt,oneside]{amsart}

\usepackage{amsmath,amsthm,amssymb}
\usepackage{verbatim}
\usepackage{fullpage}
\usepackage{graphicx}
\usepackage{wrapfig}

\newtheorem{proposition}{Proposition}
\newtheorem{corollary}{Corollary}
\newtheorem{theorem}{Theorem}
\newtheorem{lemma}{Lemma}
\newtheorem{remark}{Remark}
\newtheorem{question}{Question}
\DeclareMathOperator{\const}{const}

\begin{document}
\title{Growth of bilinear maps}
\author{Vuong Bui}
\address{Vuong Bui, Institut f\"ur Informatik, 
Freie Universit\"{a}t Berlin, Takustra{\ss}e~9, 14195 Berlin, Germany}
\thanks{The author is supported by the Deutsche 
Forschungsgemeinschaft (DFG) Graduiertenkolleg ``Facets of Complexity'' 
(GRK 2434).}
\email{bui.vuong@fu-berlin.de}
\begin{abstract}
For a bilinear map $*:\mathbb R^d\times \mathbb R^d\to \mathbb R^d$ with nonnegative coefficients and a vector $s\in \mathbb R^d$ of positive entries, among an exponential number of ways combining $n$ instances of $s$ using $n-1$ applications of $*$ for a given $n$, we are interested in the largest entry over all the resulting vectors. An asymptotic behavior is that the $n$-th root of this largest entry converges to a growth rate $\lambda$ when $n$ tends to infinity. In this paper, we prove the existence of this limit by a special structure called linear pattern. We also pose a question on the possibility of a relation between the structure and whether $\lambda$ is algebraic.
\end{abstract}
\maketitle
\section{Introduction}
Given a binary operation $*$ and an operand $s$, we have a variety of ways to combine $n$ instances of $s$ using $n-1$ applications of $*$. The results may vary as the operation $*$ is not necessarily commutative or associative. However, we might still expect that the ``largest  value'' of all the combinations does not grow too arbitrarily. A problem of this type was posed in \cite{rote2019maximum} by G\"unter Rote, where $*$ is a bilinear map with nonnegative coefficients and $s$ is a vector of positive entries, both in the same vector space. In this paper, the largest entry of all resulting vectors will be shown to be of exponential order with a fixed growth rate.

Consider a vector $s\in \mathbb R^d$ of \emph{positive} entries $s_i$ and a bilinear map $*: \mathbb R^d\times \mathbb R^d\to \mathbb R^d$ represented by \emph{nonnegative} coefficients $c_{i,j}^{(k)}$ in the way: If $v=u*w$ then $v_k=\sum_{i,j} c_{i,j}^{(k)} u_i w_j$.

Let $A_n$ for an integer $n\ge 1$ be the set of all possible vectors obtained by applying $n-1$ instances of $*$ to $n$ instances of $s$, that is $A_1=\{s\}$ and
\[
A_n = \bigcup_{1\le m\le n-1} \{x*y: x, y \in A_m \times A_{n-m}\}.
\]

For example, $A_2 = \{s*s\}$, $A_3= \{s*(s*s), (s*s)*s\}$ and $A_4= \{s*(s*(s*s)), s*((s*s)*s), (s*s)*(s*s), (s*(s*s))*s, ((s*s)*s)*s\}$. The number of combinations for $n$ is actually $C_{n-1}$, where $C_n$ is the $n$-th Catalan number. However, two different combinations may give the same result.

Let $g(n)$ denote the largest entry over all vectors in $A_n$, that is
\[
g(n) = \max \{x_i: x\in A_n, 1\le i\le d\}.
\]

For later convenient usage, we also denote by $g_i(n)$ the largest $i$-th entry over all vectors in $A_n$, that is
\[
g_i(n)=\max\{x_i: x\in A_n\}.
\]

We call the pair $(*,s)$ a \emph{system} and the following limit $\lambda$ the \emph{growth rate} of the system:
\[
\lambda = \lim_{n\to\infty} \sqrt[n]{g(n)}.
\]

We will prove the validity of this limit and give further discussions after introducing some definitions related to a special structure called \emph{linear pattern}.

Note that all the above definitions $A_n, g(n), g_i(n), \lambda$ together with some other new definitions $M(P), \lambda_P$ below depend on the considered system $(*,s)$.
Unless stated otherwise, these terms should be understood with the system given by the context.

Suppose we are given a rooted (plane) binary tree with $n$ leaves and each leave is assigned a vector. 
We can obtain a vector for any subtree from the following computation: If the subtree is just a leaf then the result is the vector assigned to that leaf; otherwise, the result is $x*y$ where $x,y$ are the results corresponding to the left and right branches, respectively. We are interested in the result for the whole tree itself, which is called the \emph{vector associated} with that tree.
If the $n$ leaves from left to right are assigned vector variables, then there is a bijection between the set of binary trees and the set of results obtained from the above process. When we fix a vector constant $s$ for each leaf, we obtain $A_n$ as the set of the results, but we may also lose injectivity in the same time.
Although there may be more than one tree giving the same $v\in A_n$, the techniques used in the work are independent of the assignment of the tree to be associated with $v$.

We call a pair of a tree $T$ with at least $2$ leaves and a \emph{marked leaf} $\ell$ of $T$ \emph{a linear pattern} $P=(T,\ell)$. This definition has some interesting properties.

\begin{wrapfigure}{l}{0.43\textwidth}
    \centering
    \includegraphics[width=0.18\textwidth]{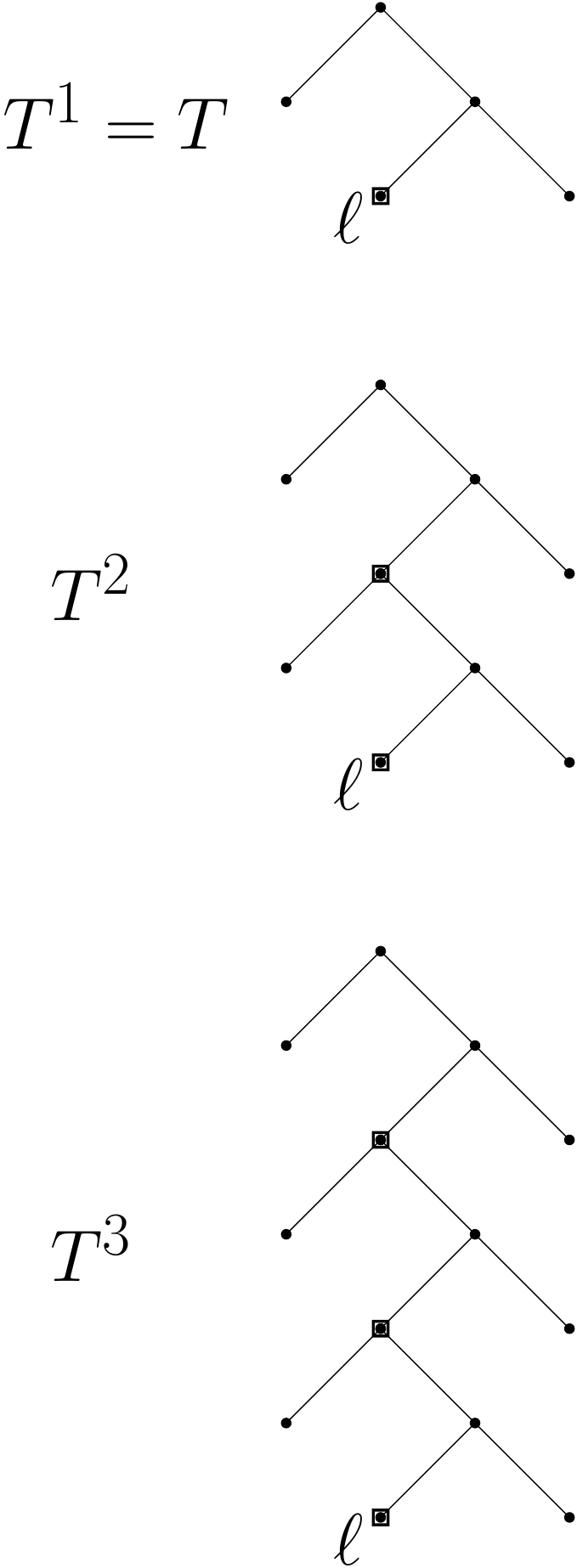}
\caption{The sequence generated by the pattern $(T,\ell)$}
\label{fig:pattern}
\end{wrapfigure}

\begin{proposition}
  \label{prop:lin-relation}
  Given a linear pattern $P=(T,\ell)$, if in the computation corresponding to $T$, the value of the marked leaf is a vector variable $u$ instead of the fixed vector $s$, then instead of a fixed result we obtain a vector $v$ related to $u$ by a matrix $M=M(P)$ such that
  \[
  v=Mu.
  \]
\end{proposition}

This fact follows from a property of bilinear maps: If we fix one of the two terms of the input, the new map will be linear. In other words, the two functions $*^L_y(x) = x*y$ and $*^R_x(y) = x*y$ are both linear. The matrix $M$ can be then constructed in a bottom-up strategy. For example, consider the pattern $(T,\ell)$ in Figure \ref{fig:pattern}, the usual resulting vector is $s*(s*s)$. If we allow the marked leaf to take a vector variable $u$, then the resulting vector $v$ becomes $s*(u*s)$. Since $*^L_s$ and $*^R_s$ are both linear, let $M_L, M_R$ be the matrices associated with $*^L_s, *^R_s$, respectively. We can see that $s*(u*s) = *^R_s(*^L_s(u)) = M_R (M_L u)$, hence $v=Mu$ for $M=M_RM_L$. More manipulations of this type can be found in Section \ref{sec:no-pattern-ex}. Such a matrix $M$ is called the \emph{associated matrix} with pattern $P$. 

A sequence of trees $\{T^t\}_{t\ge 1}$ is said to be generated by a pattern $(T, \ell)$ if $T^1=T$ and $T^t$ for $t\ge 2$ is obtained from $T$ by replacing $\ell$ by $T^{t-1}$ (see Figure \ref{fig:pattern} for example). When we address the marked leaf of $T^t$, we mean the marked leaf of the deepest instance of $T$ embedded in $T^t$.

\begin{proposition}
\label{prop:spectral-radius}
  For a linear pattern $P=(T,\ell)$, let $h(t)$ be the largest entry of the vector associated with the tree $T^t$, then the limit
  \[
  \lambda_P = \lim_{t\to\infty} \sqrt[t]{h(t)}
  \]
  is valid and equal to the spectral radius of the matrix $M(P)$.
\end{proposition}

Let $v_t$ be the vector associated with $T^t$, then $v_t=M^t s$ as a consequence of Proposition \ref{prop:lin-relation}. Therefore, Proposition \ref{prop:spectral-radius} can be deduced from Gelfand's formula: For every (not necessarily nonnegative) complex matrix $A$, $\lim_{n\to\infty} \sqrt[n]{\|A^n\|} = \rho(A)$, where $\rho(A)$ is the spectral radius and $\|A^n\|$ is any matrix norm.

The tree $T^t$ has $t(|T|-1)+1$ leaves, where $|T|$ is the number of leaves of $T$. While $h(t)$ is a lower bound for a subsequence of $g(n)$, the corresponding lower bound for the growth rate should be the $(|T|-1)$-th root $\bar\lambda_P$ of $\lambda_P$ instead. We call $\bar\lambda_P$ the \emph{rate of pattern} $P$.

\begin{proposition}
\label{prop:lower-bound}
  For every linear pattern $P=(T,\ell)$,
  \[
  \liminf_{n\to\infty} \sqrt[n]{g(n)} \ge \bar\lambda_P.
  \]
\end{proposition}

Indeed, let $m=|T|-1$, for each $n=mp+r$ ($1\le r\le m$), consider the tree with $n$ leaves obtained from $T^p$ by replacing the marked leaf by any tree $T_0$ with $r$ leaves. Since the associated vector with $T_0$ is positive and bounded, it is not hard to see that the $n$-th root of the largest entry of the vectors associated with these trees converges to $\bar\lambda_P$.

Moreover, we give the following stronger conclusion, which confirms the validity of $\lambda$.

\begin{theorem}
  \label{thm:pattern}
The $n$-th root of $g(n)$ converges when $n$ tends to infinity and the limit is the supremum of $\bar\lambda_P$ over all patterns $P$, that is
\[
\lambda=\lim_{n\to\infty} \sqrt[n]{g(n)} = \sup_P \bar\lambda_P.
\]
\end{theorem}

The growth rate $\lambda$ is said to be \emph{recognized} by a pattern if the rate of the pattern is $\lambda$. There exist some cases for which no pattern recognizes $\lambda$. The following system is one example.

\begin{theorem}
\label{thm:no-pattern-ex}
If $s=(1,1)$ and
\[
x*y = (x_1 y_1 + x_2 y_2, x_2 y_2),
\]
then $\lambda > \bar\lambda_P$ for every $P$.
\end{theorem}

For this system, the value of $g(n)$ can be found in the vectors associated with the perfect binary trees (for $n$ being a power of $2$), which cannot be generated by any linear pattern.
Actually the system in the above theorem was studied in a different formulation (see \cite{de2012maximum}) and the growth rate was shown to be
\[
\lambda = \exp(\sum_{i\ge 1} \frac{1}{2^i} \log(1+\frac{1}{a_i^2})) = 1.502836801\dots,
\]
where $a_n$ is the sequence with $a_0=1$ and $a_{k+1}=1+a_k^2$ for $k\ge 0$.

This constant has also been studied as the rates of quadratic recurrences and $a_n=g(2^n)$ is the number of binary trees of heights at most $n+1$ (Sequence $A003095$, The On-Line Encyclopedia of Integer Sequences). For more information on this and other sequences of the type, see \cite{aho1973some}.

Suppose the coefficients of $*$ and the entries of $s$ in a system are all integers, the entries of the matrix $M=M(P)$ for any pattern $P$ are also integers as one can see from the construction of $M$ (described after Proposition \ref{prop:lin-relation}). It follows that the spectral radius of $M$ is algebraic.
This means the growth rate is algebraic whenever a pattern recognizes it. Since the growth rate $\lambda$ in Theorem \ref{thm:no-pattern-ex} seems to be not an algebraic number, it suggests the following question of the other direction.

\begin{question}
Suppose both the coefficients of $*$ and the entries of $s$ are integers. Is it true that: If $\lambda$ is algebraic, then there exists a pattern $P$ such that $\bar\lambda_P = \lambda$?
\end{question}

It makes sense to give an example where a pattern recognizes the growth rate, and hence, the growth rate is algebraic.

\begin{theorem}
\label{thm:fibo}
If $s=(1,1)$ and
\[
x*y = (x_1 y_2 + x_2 y_1, x_1 y_2),
\]
then the growth rate $\lambda$ is the golden ratio $\phi$, which is recognized by a pattern.
In particular, $g_1(n)=F_{n+1}$ and $g_2(n)=F_n$, where $F_n$ is the Fibonacci sequence with $F_1=F_2=1$.
\end{theorem}

In this system, the value of $g(n)$ can be found in the vectors associated with the binary trees where the right branch of every non-leaf vertex is just a leaf. The tree of the pattern has only two leaves with the marked leaf on the left. The proof uses some inequalities involving the elements of the Fibonacci sequence that are interesting on their own.

The readers may notice that although the two examples in Theorem \ref{thm:no-pattern-ex} and Theorem \ref{thm:fibo} just slightly differ from each other, the growth rates and the patterns are quite different in nature.

Further discussions will be given in Section \ref{sec:discussions}. In particular, we explain why we require the signs of the coefficients and the entries. In short, it becomes too trivial if we restrict the requirement of $*$, and the convergence does not always hold if the requirements of $*$ and $s$ are relaxed.
We also sketch how the terms in this work look like in the problem of \cite{rote2019maximum}.   Another example is also introduced as an open problem there. Discussions of the possibility using other norms than the maximum norm of $g(n)$ and multilinear maps rather than bilinear maps are also given.

We give the proofs of Theorems \ref{thm:pattern}, \ref{thm:no-pattern-ex} and \ref{thm:fibo} in Sections \ref{sec:pattern}, \ref{sec:no-pattern-ex} and \ref{sec:fibo}, respectively. Section \ref{sec:lemmas} proves the lemmas used in Section \ref{sec:pattern}. The proofs of Propositions \ref{prop:lin-relation}, \ref{prop:spectral-radius} and \ref{prop:lower-bound} can be reproduced from the explanation given after each of them, therefore, they are omitted.
\section{Proof of Theorem \ref{thm:pattern}}
\label{sec:pattern}

Consider the \emph{dependency graph} that is a directed graph whose vertices are $1,\dots,d$; there is a directed edge from $k$ to $i$ if and only if there exists some $j$ such that $c_{i,j}^{(k)}$ or $c_{j,i}^{(k)}$ is positive, where $c_{i,j}^{(k)}$ are the coefficients of $*$.
We say $k$ depends on $i$ for such an edge $ki$. In some cases, we need to say specifically that $k$ \emph{left depends} (resp. \emph{right depends}) on $i$ if $c_{i,j}^{(k)}$ (resp. $c_{j,i}^{(k)}$) is positive.

The dependency graph can be partitioned into strongly connected components, which can be partially ordered. For different components $C_1, C_2$, we say $C_1$ is greater than $C_2$ if either there is a directed edge $ij$ for $i\in C_1, j\in C_2$, or there exists another component $C_3$ so that $C_1 > C_3 > C_2$.

For a given component $C$, consider the subgraph induced by all the components smaller than or equal to C. We define the $C$-subsystem as the system $(*', s')$ induced by the dimensions corresponding to the vertices in the subgraph. In other words, let $d'$ be the number of vertices $d'$ in the subgraph, then $*'\in \mathbb R^{d'}\times\mathbb R^{d'}\to\mathbb R^{d'}$ and $s'\in \mathbb R^{d'}$ have correspondingly the coefficients and the entries from the original system. Note that if $i$ is a dimension in the $C$-subsystem, then the $i$-th entry of the resulting vector for every combination is the same for both the $C$-subsystem and the original system.

To illustrate the definitions, consider the case $s=(1,1)$ and $(x*y)=(x_1 y_2 + x_2 y_1, x_2 y_2)$.

The dependency graph contains two vertices $1,2$ and three edges: a loop at $1$, a loop at $2$ and an edge from $1$ to $2$. There are two strongly connected components $C_1$ and $C_2$, each $C_i$ containing only vertex $i$. We also have the order $C_1>C_2$.
While the $C_1$-subsystem is actually the original system, the $C_2$-subsystem is in a $1$-dimensional space where $s'$ is the unit scalar and $*'$ is the usual product of numbers.

This example is actually interesting on its own since every combination gives the same result $(n,1)$. Together with the two examples in Theorem \ref{thm:no-pattern-ex} and Theorem \ref{thm:fibo}, these three examples are very similar by notation but so different in nature.

Before proving the theorem, we first give some useful lemmas, which will be proved later in Section \ref{sec:lemmas}.
\begin{lemma}
\label{lem:bounded-ratio-adj}
For every $i$, the value $g_i(n)$ is at least a constant times $g_i(n+1)$.
\end{lemma}
\begin{lemma}
\label{lem:comparing-lim-inf-sup}
If $i,j$ are in the same component, then
\begin{align*}
\liminf_{n\to\infty} \sqrt[n]{g_i(n)} &= \liminf_{n\to\infty} \sqrt[n]{g_j(n)},\\
\limsup_{n\to\infty} \sqrt[n]{g_i(n)} &= \limsup_{n\to\infty} \sqrt[n]{g_j(n)}.
\end{align*}

If $i\in C_1$, $j\in C_2$ and $C_1< C_2$ then
\begin{align*}
\liminf_{n\to\infty} \sqrt[n]{g_i(n)} &\le \liminf_{n\to\infty} \sqrt[n]{g_j(n)},\\
\limsup_{n\to\infty} \sqrt[n]{g_i(n)} &\le \limsup_{n\to\infty} \sqrt[n]{g_j(n)}.
\end{align*}
\end{lemma}
\begin{lemma}
\label{lem:same-comp-make-pattern}
  Given a pattern $P=(T,\ell)$ with the associated matrix $M$. Let $i,j$ be two vertices of the same component, then there exists a pattern $P'=(T',\ell')$ with the difference in the number of leaves $|T'|-|T|$ bounded and $\lambda_{P'}$ at least a constant times $M_{i,j}$.
\end{lemma}

\begin{lemma}
\label{lem:entries-not-too-big}
If $M=M(P)$ is the matrix associated with a pattern $P=(T,\ell)$ with $T$ having $n$ leaves, then for every $i,j$, the value $M_{i,j}$ is at most a constant times $g_i(n)$.
\end{lemma}

\begin{lemma}
\label{lem:g_i(n)-large}
If a component $C$ is greater than every other component, then $g_i(n)$ is at least a constant times $g(n)$ for every $i\in C$.
\end{lemma}

\begin{lemma}
\label{lem:tree-division}
For every binary tree with $n>1$ leaves, there is a subtree with $m$ leaves such that $n/3 \le m \le 2n/3$.
\end{lemma}

We are now ready to prove Theorem \ref{thm:pattern}.

For a component $C$, let $\lambda_P^{C}$ and $\bar\lambda_P^{C}$ denote the rates with respect to the $C$-subsystem.

Take any component $C$, we investigate the $C$-subsystem. This restriction actually does not reduce the generality but allows us to conclude on the convergence of $\sqrt[n]{g_i(n)}$ for every $i$, as we will show later.

Suppose $C$ is a component $C_0$ such that
\begin{equation}
\label{eq:requirement-of-C_0}
\limsup_{n\to\infty} \sqrt[n]{g_i(n)} > \limsup_{n\to\infty} \sqrt[n]{g_j(n)}
\end{equation}
for every $C'<C_0, i\in C_0, j\in C'$.

It can be seen that $\liminf_{n\to\infty} \sqrt[n]{g_i(n)} \ge \sup_P \bar\lambda_P^{C_0}$ for every $i\in C_0$ by Proposition \ref{prop:lower-bound} and Lemma \ref{lem:g_i(n)-large}. With the condition of $C_0$, we prove the other direction: For every $i\in C_0$,
\[
\limsup_{n\to\infty} \sqrt[n]{g_i(n)} \le \sup_P \bar\lambda_P^{C_0}.
\]

For a vertex $k$, denote $\theta_k = \limsup_{n\to\infty} \sqrt[n]{g_k(n)}$.
By definition, for every $\epsilon>0$, there exists an $n_{k,\epsilon}$ such that for every $n>n_{k,\epsilon}$, we have $g_k(n) < (\theta_k+\epsilon)^n $. Also, for every $N$, there exists $n>N$ such that $g_k(n)>(\theta_k-\epsilon)^n$.

Let $i$ be a vertex in $C_0$ and denote $\theta=\theta_i$.

Fix $\epsilon$, choose $n_\epsilon=\max_k n_{k,\epsilon}$. Let $N=3n_\epsilon$ and take any $n>N$ such that $g_i(n) > (\theta-\epsilon)^n$.

Let $T$ be a tree so that the $i$-th entry of the associated vector is $g_i(n)$. Take a subtree $T_2$ with $m$ leaves so that $n/3\le m\le 2n/3$ (by Lemma \ref{lem:tree-division}), and combine any leaf $\ell_2$ among the $m$ leaves with $T_2$ to obtain the pattern $P_2 = (T_2,\ell_2)$. Denote by $\ell_1$ the root of $T_2$, and by $T_1$ the tree obtained from $T$ after contracting $T_2$ to $\ell_1$. We have another pattern $P_1=(T_1,\ell_1)$. Also, consider the pattern $P=(T,\ell)$ for $\ell=\ell_2$.

Let the matrices for $P,P_1,P_2$ be $M,A,B$, respectively. Clearly,
\[
M=AB.
\]

Since $g_i(n)=\sum_j M_{i,j} s_j$, there exists some $j$ and a constant $\alpha>0$ such that 
\[
M_{i,j} \ge \alpha g_i(n).
\]

Since $M_{i,j} = \sum_k A_{i,k} B_{k,j}$, there exists some $k$ and a constant $\beta>0$ such that 
\[
A_{i,k} B_{k,j} \ge \beta M_{i,j}\ge \beta\alpha g_i(n) > \beta\alpha (\theta-\epsilon)^n.
\]

Denote $\theta'=\theta_k$. By Lemma \ref{lem:entries-not-too-big}, and by the definition of $\theta'$ with $m>n_\epsilon$, there exists a constant $\gamma>0$ such that 
\[
B_{k,j} \le \gamma g_k(m) < \gamma (\theta'+\epsilon)^m.
\]

It means 
\[
A_{i,k} > \frac{\alpha\beta}{\gamma} \frac{(\theta-\epsilon)^n}{(\theta'+\epsilon)^m}.
\]

Note that $\frac{\alpha\beta}{\gamma}$ is a constant and 
\begin{align*}
\frac{(\theta-\epsilon)^n}{(\theta'+\epsilon)^m} &= (\frac{\theta-\epsilon}{\theta'+\epsilon})^m (\theta-\epsilon)^{n-m}\\
 &\ge (\sqrt{\frac{\theta-\epsilon}{\theta'+\epsilon}})^{n-m} (\theta-\epsilon)^{n-m}\\
&= (\sqrt{\frac{\theta-\epsilon}{\theta'+\epsilon}}\frac{\theta-\epsilon}{\theta+\epsilon})^{n-m} (\theta+\epsilon)^{n-m},
\end{align*}
where the inequality step is due to $m\ge (n-m)/2$.

Suppose $k$ is in a smaller component than $C_0$, that is $\theta'<\theta$.

When $\epsilon$ is small and $n$ is large enough, the value of $A_{i,k}$ will not be bounded by a constant times $(\theta+\epsilon)^{n-m}$ due to $(\theta-\epsilon)/(\theta'+\epsilon) > 1$ but $(\theta-\epsilon)/(\theta+\epsilon)$ tending to $1$ when $\epsilon$ tends to $0$. However, $A_{i,k}\le \zeta g_i(n-m+1)\le \zeta \eta g_i(n-m) < \zeta\eta (\theta+\epsilon)^{n-m}$, where $\zeta, \eta$ are the constants obtained respectively from Lemma \ref{lem:entries-not-too-big} and Lemma \ref{lem:bounded-ratio-adj}, a contradiction (note that $T_1$ has $n-m+1$ leaves).

Therefore, $i$ and $k$ are in the same component when $n$ is large enough, which means $B_{k,j}$ is at most a constant times $(\theta+\epsilon)^m$. It follows that $A_{i,k}$ is at least a constant times
\[
\frac{(\theta-\epsilon)^{n}}{(\theta+\epsilon)^m} .
\]

For every $\epsilon'>0$, there exists $\epsilon>0$ such that 
\[
\frac{(\theta-\epsilon)^n}{(\theta+\epsilon)^m} > (\theta-\epsilon')^{n-m}.
\]

By Lemma \ref{lem:same-comp-make-pattern}, the lower bound of $A_{i,k}$ means that for every $\epsilon''>0$, there exists a pattern $P'$ having $\bar\lambda_{P'}^{C_0} > \theta-\epsilon''$ (by setting $\epsilon'$ small enough). In other words,
\[
\limsup_{n\to\infty} \sqrt[n]{g_i(n)} \le \sup_P \bar\lambda_P^{C_0}.
\]

It means $\lim_{n\to\infty} \sqrt[n]{g_i(n)}$ exists for every $i\in C_0$ since the limit superior and the limit inferior are equal.

We have shown that $\sqrt[n]{g_i(n)}$ converges to a limit for every $i$ in a component satisfying the requirement \eqref{eq:requirement-of-C_0}. It remains to consider components $C$ not satisfying the requirement. For such a component $C$, there is a component $C_0 < C$ satisfying that requirement and $\limsup_{n\to\infty} \sqrt[n]{g_i(n)} = \limsup_{n\to\infty} \sqrt[n]{g_k(n)}$ for any $i\in C_0$ and $k\in C$. By Lemma \ref{lem:comparing-lim-inf-sup},

\[
\liminf_{n\to\infty} \sqrt[n]{g_k(n)} \ge \liminf_{n\to\infty} \sqrt[n]{g_i(n)} = \limsup_{n\to\infty} \sqrt[n]{g_i(n)} =\limsup_{n\to\infty} \sqrt[n]{g_k(n)}.
\]

It means $\lim_{n\to\infty} \sqrt[n]{g_k(n)}$ exists because the limit superior and the limit inferior are equal.

The existence of 
\[
\lambda = \lim_{n\to\infty} \sqrt[n]{g(n)} = \max_k \lim_{n\to\infty} \sqrt[n]{g_k(n)}
\]
follows from the existence of $\lim_{n\to\infty} \sqrt[n]{g_k(n)}$ for every $k$. 

This limit $\lambda$ is equal to the supremum of $\bar\lambda_P$ over all patterns $P$ because for $i\in C$ satisfying $\lim_{n\to\infty} \sqrt[n]{g_i(n)} = \lambda$, we have
\[
\sup_P \bar\lambda_P^{C} \le \sup_P \bar\lambda_P \le \lim_{n\to\infty} \sqrt[n]{g(n)} = \lim_{n\to\infty} \sqrt[n]{g_i(n)} = \sup_P \bar\lambda_P^{C}.
\]

\section{Proofs of the lemmas}
\label{sec:lemmas}
\begin{proof}[Proof of Lemma \ref{lem:bounded-ratio-adj}]
Let $T$ be a tree with $n+1$ leaves so that the $i$-th entry of the associated vector is $g_i(n+1)$. Take any subtree $T_0$ with $2$ leaves, and replace it by a leaf, denoted by $\ell$, to obtain a new tree $T'$ with $n$ leaves.

Let $v, v'$ be the vector associated with the trees $T, T'$, respectively.

Let $M$ be the matrix associated with the pattern $(T',\ell)$, that is $v'=Ms$ for the vector $s$ associated with the leaf $\ell$. If the leaf $\ell$ is replaced by the tree $T_0$, we have the relation $v=Mu$ for the vector $u=s*s$ associated with $T_0$.

Since $u_i\le g(2)$ and $s_i\ge \min_k s_k$ for every $i$,
\[
\frac{u_i}{s_i} \le \frac{g(2)}{\min_k s_k}.
\]

Together with $v=Mu$ and $v'=Ms$, we have
\[
\frac{v_i}{v'_i} \le \frac{g(2)}{\min_k s_k}.
\]

The conclusion follows due to $v'_i \le g_i(n)$ and $v_i=g_i(n+1)$.
\end{proof}
\begin{remark}
It is possible to obtain a more general conclusion by choosing $T_0$ with more than two leaves. However, we cannot guarantee the size of $T_0$ in this case but only some bounds on it (e.g. Lemma \ref{lem:tree-division}). The question is: Is it true that $g(p+q)\le\const g(p)g(q)$ for every $p,q\ge 1$? If this is true, not only the validity of $\lambda$ just follows but we can also conclude that $g(n) \ge \const \lambda^n$ (by Fekete's lemma \cite{fekete1923verteilung}).
\end{remark}
The proofs of the remaining lemmas use the following obvious corollary of Lemma \ref{lem:bounded-ratio-adj}.
\begin{corollary}
\label{cor:any-distance}
Given a fixed $\delta$, for every $i$, the value $g_i(n)$ is at least a constant times $g_i(n+\delta)$.
\end{corollary}

\begin{proof}[Proof of Lemma \ref{lem:comparing-lim-inf-sup}]
Suppose there is an edge $ki$ in the dependency graph. For each $n$, let $T$ be a tree so that the $i$-th entry of the associated vector is $g_i(n)$. Consider the tree $T'$ with $n+1$ leaves where the left (resp. right) branch of the root is $T$ if $k$ left (resp. right) depends on $i$, and the other branch is just a single leaf. It can be seen from $T'$ that there exists a constant $\alpha>0$ such that
\[
g_k(n+1)\ge \alpha g_i(n).
\]

Let $i,j$ be two vertices so that there exists a path of length $d$ from $j$ to $i$, there exists a constant $\beta>0$ such that
\[
g_j(n+d)\ge \beta g_i(n).
\]

Since $g_j(n)\ge \gamma g_j(n+d)$ where $\gamma>0$ is the constant obtained from Corollary \ref{cor:any-distance}, we have
\begin{equation}
\label{eq:graph->value-comparision}
g_j(n) \ge \gamma g_j(n+d) \ge \gamma\beta g_i(n).
\end{equation}

If there is a path from $j$ to $i$, it follows from Equation \eqref{eq:graph->value-comparision} that
\begin{align*}
\liminf_{n\to\infty} \sqrt[n]{g_i(n)} &\le \liminf_{n\to\infty} \sqrt[n]{g_j(n)},\\
\limsup_{n\to\infty} \sqrt[n]{g_i(n)} &\le \limsup_{n\to\infty} \sqrt[n]{g_j(n)}.
\end{align*}
It is indeed the case when $i\in C_1, j\in C_2$ and $C_1<C_2$.

If $i, j$ are in the same component, then there exist a path from $i$ to $j$ and also a path from $j$ to $i$. Apply the above inequalities to both $i,j$ and $j,i$, we obtain
\begin{align*}
\liminf_{n\to\infty} \sqrt[n]{g_i(n)} &= \liminf_{n\to\infty} \sqrt[n]{g_j(n)},\\
\limsup_{n\to\infty} \sqrt[n]{g_i(n)} &= \limsup_{n\to\infty} \sqrt[n]{g_j(n)}.
\end{align*}
\end{proof}
\begin{proof}[Proof of Lemma \ref{lem:same-comp-make-pattern}]
We can assume $i \ne j$, since otherwise, we just set $P'=P$ as $\lambda_P\ge M_{i,i}$.

Since $i,j$ are two distinct vertices in the same component, there is always a path from $j$ to $i$. Let the path be $k_0,\dots,k_d$, where $d$ is the length of the path and $k_0=j$, $k_d=i$. Construct the trees $T_0,\dots,T_d$ such that $T_d=T$, and for $t<d$, one of the two branches of the root of $T_t$ is $T_{t+1}$ and the other is just a single leaf. If $k_t$ left (resp. right) depends on $k_{t+1}$ then the branch of $T_{t+1}$ is on the left (resp. right) in $T_t$.

Let $P'$ be the pattern $(T',\ell')$ for $T'=T_0$, $\ell'=\ell$, and $M'$ the matrix associated with $P'$. We can see that $|T'|-|T|$ is bounded and $M'_{j,j}$ is at least a constant times $M_{i,j}$. It follows that $\lambda_{P'}$ is at least a constant times $M_{i,j}$ since $\lambda_{P'}\ge M'_{j,j}$.
\end{proof}

\begin{proof}[Proof of Lemma \ref{lem:entries-not-too-big}]
Let the vector associated with $T$ be $v=Ms$. Since $v_i=\sum_j M_{i,j} s_j$ and $v_i\le g_i(n)$, the value $M_{i,j}$ for any $j$ is at most a constant times $g_i(n)$.
\end{proof}
\begin{proof}[Proof of Lemma \ref{lem:g_i(n)-large}]
Since there is a path from $i$ to $j$ for every $i\in C$ and any other $j$, by the same argument as in Equation \eqref{eq:graph->value-comparision}, there exists a constant $\alpha_j>0$ such that
\[
g_i(n)\ge\alpha_j g_j(n).
\]

If we do not fix $j$, let $\alpha=\min_j \alpha_j$, we have
\[
g_i(n)\ge \max_j \alpha_j g_j(n) \ge \alpha\max_j g_j(n) = \alpha g(n).
\]
\end{proof}

\begin{proof}[Proof of Lemma \ref{lem:tree-division}]
This fact is well known and its easy verification is left to the readers.
\end{proof}

\section{Proof of Theorem \ref{thm:no-pattern-ex}}
\label{sec:no-pattern-ex}
Consider a pattern $P=(T,\ell)$ with its matrix
\[
\begin{bmatrix}
  a&b\\
  c&d
\end{bmatrix}.
\]

It is verifiable that $a \ge 1, b \ge 1, c = 0$ and $d = 1$ (the readers can check for themselves, e.g. by induction through the manipulations of patterns and matrices throughout the proof).
The spectral radius of the matrix can be also seen to be $a$. Therefore, the rate is the $(|T|-1)$-th root of $a$, where $|T|$ is the number of leaves in $T$.

Consider some two patterns $P_1=(T_1,\ell_1)$ and $P_2=(T_2,\ell_2)$ with their associated matrices respectively
\[
\begin{bmatrix}
  a_1&b_1\\
  0&1
\end{bmatrix},
\begin{bmatrix}
  a_2&b_2\\
  0&1
\end{bmatrix}.
\]
Their product is
\[
\begin{bmatrix}
  a_1 a_2 & a_1 b_2 + b_1\\
  0&1
\end{bmatrix}
,
\]
which is the matrix associated with the pattern $P=(T,\ell)$ with $T$ obtained from $T_1$ by replacing $\ell_1$ by $T_2$ and letting $\ell$ be $\ell_2$.

We have
\begin{equation}
\label{eq:at-most-its-sub}
\bar\lambda_P \le\max\{\bar\lambda_{P_1},\bar\lambda_{P_2}\},
\end{equation}
since $|T|-1=(|T_1|-1)+(|T_2|-1)$ and the spectral radius of the product is $a_1 a_2$.

Suppose there is a pattern $P$ with $\bar\lambda_{P}=\lambda$, let $P^*=(T^*,\ell^*)$ be a pattern with the minimal number of leaves in the tree among all such patterns. By Equation \eqref{eq:at-most-its-sub}, we can see that $P^*$ is not decomposable into two patterns in that way. In other words, one branch of the root of $T^*$ is just the marked leaf $\ell^*$.

Let the other branch than the branch of the marked leaf, denoted by $T'$, have the associated vector $(a, 1)$, then the matrix associated with $P^*$ is
\[
\begin{bmatrix}
  a&1\\
  0&1
\end{bmatrix}
.
\]

We have $\bar\lambda_{P^*}=\sqrt[m]{a}$, where $m$ is the number of leaves in $T'$.

Let $T''$ be a tree where each branch of the root is a copy of $T'$. The vector associated with $T''$ is $(a^2 + 1, 1)$. Since $\sqrt[2m]{a^2+1} > \sqrt[m]{a}$, if we replace $T'$ in $T^*$ by $T''$, we obtain another pattern with a higher rate than $\bar\lambda_{P^*}$, a contradiction.
\begin{remark}
A corollary from the proof is $\lambda=\sup_n \sqrt[n]{g(n)}$. This can be also obtained by Fekete's lemma since $g(n)=g_1(n)$ and $g_1(p+q)\ge g_1(p)g_1(q)$ for $p,q\ge 1$.
\end{remark}
\section{Proof of Theorem \ref{thm:fibo}}
\label{sec:fibo}
It can be seen that $g_1(n)\ge F_{n+1}$ and $g_2(n)\ge F_n$ for every $n\ge 1$ since the vector $(F_{n+1}, F_n)$ is associated with the tree $T^{n-1}$ for the pattern $(T,\ell)$ where $T$ is the tree with two leaves, the marked leaf $\ell$ is the one on the left.

In order to show that they are also the upper bounds, we prove the following lemma.
\begin{lemma}
Let $\{F_n\}_{n\ge 0}$ be the Fibonacci sequence with $F_0=0, F_1 = 1, F_2=1$, then the following inequalities
\begin{align*}
  F_{p} F_{q-1} + F_{p-1} F_{q} &\le  F_{p+q-1},\\
  F_p F_q &\le F_{p+q-1}
\end{align*}
hold for every $p,q\ge 1$
\end{lemma}

\begin{proof}
The conclusion holds for any $(p,q)\in  (\{1,2\} \times \mathbb N^+ )\cup (\mathbb N^+ \times \{1,2\})$, i.e. one of the four conditions $p=1$, $p=2$, $q=1$, $q=2$ holds.

For the first inequality, if $p=1$ (similarly for $q=1$), then the inequality is equivalent to $F_{q-1} \le F_q$. If $p=2$ (similarly for $q=2$), then it is equivalent to $F_{q-1}+F_q \le F_{q+1}$.

For the second inequality, if $p=1$ (similarly for $q=1$), then the inequality is equivalent to $F_q \le F_q$. If $p=2$ (similarly for $q=2$), then it is equivalent to $F_q \le F_{q+1}$. 

We prove the lemma by induction. Suppose the inequalities hold for any $(p', q') \in \{p-1, p-2\} \times \{q-1, q-2\}$, we show that they also hold for $(p,q)$. 

Indeed,
\begin{align*}
  F_{p} F_{q-1} + F_{p-1} F_{q} &= (F_{p-2}+F_{p-1})(F_{q-3}+F_{q-2}) + (F_{p-3}+F_{p-2})(F_{q-2}+F_{q-1})\\
  &= (F_{p-2} F_{q-3}+F_{p-3} F_{q-2})+(F_{p-2} F_{q-2}+F_{p-3} F_{q-1})\\
  &\;\;\;\;+(F_{p-1} F_{q-3}+F_{p-2} F_{q-2})+(F_{p-1} F_{q-2}+F_{p-2} F_{q-1}) \\
  &\le F_{p+q-5} + F_{p+q-4} + F_{p+q-4} + F_{p+q-3} \\
  &= F_{p+q-3} + F_{p+q-2} \\
  &= F_{p+q-1},
\end{align*}
and
\begin{align*}
  F_p F_q &= (F_{p-2}+F_{p-1}) (F_{q-2}+F_{q-1})\\
  &= F_{p-2} F_{q-2}+F_{p-2} F_{q-1}+F_{p-1} F_{q-2}+F_{p-1} F_{q-1}\\
  &\le F_{p+q-5} + F_{p+q-4} + F_{p+q-4} + F_{p+q-3}\\
  &\le F_{p+q-3} + F_{p+q-2}\\
  &\le F_{p+q-1}.
\end{align*}

By induction, the inequalities hold for every $p,q\ge 1$.
\end{proof}

Now the verification for the upper bounds of $g_1(n)$ and $g_2(n)$ becomes clear. They hold trivially for $n=1$. For higher $n$, if $g_1(n)$ corresponds to a tree where the left branch of the root has $p$ leaves and the right branch has $q$ leaves ($p+q=n$), then the same bounds hold: 
\[
g_1(n) \le g_1(p) g_2(q) + g_2(p) g_1(q) = F_{p+1} F_q + F_p F_{q+1} \le F_{p+q+1}=F_{n+1},
\]
and 
\[
g_2(n) \le g_1(p) g_2(q) = F_{p+1} F_q \le F_{p+q}=F_n.
\]

Being both lower bounds and upper bounds, we have $g_1(n)=F_{n+1}$ and $g_2(n)=F_n$.

\section{Some further discussions}
\label{sec:discussions}
It is natural to require the coefficients of $*$ to be nonnegative and the proof of the limit becomes very simple when the coefficients are positive. Indeed, one can extract it from the original proof for the case where the dependency graph is connected, or even better by showing easily that for every $i$, the sequence $\{c^{(i)}_{i,i} g_i(n)\}_n$ is supermultiplicative (and then applying Fekete's lemma). Actually, when the dependency graph is only connected, applying the trick in the proof of Lemma \ref{lem:comparing-lim-inf-sup} on both branches, one can also show that for every $i$, the sequence $\{\alpha_i g_i(n-\delta_i)\}_n$ for some constant $\alpha_i$ and some integer $\delta_i$ is supermultiplicative. This is the approach for a second proof of the limit in this specific case. By Fekete's lemma, we have a corollary that $g(n) \le \const \lambda^n$ beside the convergence of $\sqrt[n]{g(n)}$. This is left as an exercise for the readers. 

However, it does not mean that requiring $s$ to be positive is unnecessary. For example, the following system does not have a growth rate: $s=(1,0)$ and $x*y= (x_2 y_2, x_1 y_1)$. The readers can verify that $g(n) = 0$ if $3\mid n$ and $g(n)=1$ otherwise.

In general, we cannot allow either any entry of $s$ or any coefficient of $*$ to be negative even after choosing an appropriate norm for $g(n)$, say the greatest absolute value of an entry (i.e. maximum norm). For example, if some entries of $s$ are allowed to be negative, consider the following system: $s=(1,-1,1)$ and $x*y = (x_1y_1, x_2y_2, 3x_1y_3 + 3x_2y_3)$, we have $g(n)=1$ for even $n$ and $g(n)=6^{\frac{n-1}{2}}$ for odd $n$. If some coefficients of $*$ are allowed to be negative, consider the following system: $s=(1,1,1)$ and $x*y=(x_1y_1, -x_2y_2, 3x_1y_3 - 3x_2y_3)$, we have the same $g(n)$, that is $g(n)=1$ for even $n$ and $g(n)=6^{\frac{n-1}{2}}$ for odd $n$. We leave the simple verification of these facts to the readers as an exercise. This behavior is rather different from that of matrices, where the convergence of $\sqrt[n]{\|A^n\|}$ holds for any complex matrix $A$, not necessarily nonnegative ones, by Gelfand's formula. Note that although our work is about bilinear maps, the problem mostly boils down to linear maps. 

A more sophisticated example than what have presented so far can be found in \cite{rote2019maximum} with growth rate $\sqrt[13]{95}$ and a complex linear pattern. It actually solves a problem on the maximal number of minimal dominating sets in a tree. One can also find a proof of the validity of $\lambda$ for that particular case there as the dependency graph is connected (with the trick described above). In fact, the function $g(n)$ there does not use the maximum norm but a linear combination of the entries in the resulting vectors, and the vector $s$ there has some zero entries. 

We also give an interesting example that looks more like the system in Theorem \ref{thm:no-pattern-ex} than the one in Theorem \ref{thm:fibo} but has the growth rate closer to the growth rate in Theorem \ref{thm:fibo} than the one in Theorem \ref{thm:no-pattern-ex}. Consider the system $(*,s)$ with $s=(1,1)$ and $x*y=(x_1 y_1 + x_2 y_2, x_1 y_1)$. Actually we have not been able to calculate the growth rate of this system, but an estimation by its sequence $g(n)$ would guess it is just a bit larger than the growth rate $\phi$ of Theorem \ref{thm:fibo}. In particular, $\lambda \ge \phi$ since the pattern $(T,\ell)$ with $T$ having two leaves and $\ell$ being one of them has the rate $\phi$. For the case $T$ is a tree of three leaves with one of the branches being precisely the marked leaf $\ell$ (the other branch is a tree of two leaves), the readers can check that the rate of this pattern is greater than $\phi$, hence $\lambda > \phi$. Although $\lambda$ is closer to $\phi$ than the growth rate in Theorem \ref{thm:no-pattern-ex}, the trees corresponding to $g(n)$ in this system seem to follow a symmetric pattern like in Theorem \ref{thm:no-pattern-ex} rather than a linear pattern as in Theorem \ref{thm:fibo}. We do not investigate this example here but just would like to point out this as an interesting open problem, and that the growth of bilinear maps defines a new ``language'' to describe a class of constants.

We would note that the maximum norm of $g(n)$ is chosen for convenience without loss of generality. That is because every two arbitrary norms on $\mathbb R^d$ are within a constant factor of each other. In other words, except the statements on the precise value of $g(n)$, the other statements on its asymptotic behaviors still hold with another norm.

The title of the work is about bilinear maps as they are quite popular. However, the approach still works for multilinear maps with almost no essential change. The readers can see that we did not use any specific property of bilinear maps but fixing one input and converting them into linear maps for the other input. Fixing all but one input works for multilinear maps as well. The matrix manipulation for linear maps still remains the same.
\section*{Acknowledgement}
The author would like to thank G\"unter Rote for introducing the problem and reading the proofs; Roman Karasev and especially the anonymous referee for their suggestions to various improvements in the presentation.
\bibliographystyle{unsrt}
\bibliography{gbm}

\end{document}